\documentclass{ws-rv9x6}
\usepackage{subfigure}   
\usepackage{ws-rv-thm}   
\usepackage[square]{ws-rv-van}   
\makeindex

\usepackage{proof}

\newcommand{\NN}{\mathbf{N}}
\newcommand{\QQ}{\mathbf{Q}}

\newcommand{\A}{\mathbf{C}}
\newcommand{\Aco}{\mathbf{C'}}
\newcommand{\C}{\mathbf{S}}

\newcommand{\apart}{\,{\not\,\not\hspace{-0.07em}=}\,}

\newcommand{\Amb}{\mathbf{Amb}}
\newcommand{\amb}{\mathbf{Amb}}

\newcommand{\Fun}{\mathbf{Fun}}
 
\newcommand{\Left}{\mathbf{Left}}
\newcommand{\Nil}{\mathbf{Nil}}
\newcommand{\Pair}{\mathbf{Pair}}

\newcommand{\Right}{\mathbf{Right}}

\newcommand{\Set}{{\mathbf{\downdownarrows}}}
\newcommand{\Setomega}{\Set^*}

\newcommand{\False}{\mathbf{False}}

\newcommand{\mycase}{\mathbf{case}}

\newcommand{\of}{\mathbf{of}}

\newcommand{\munu}{\Box}

\newcommand{\re}{\mathbf{r}}

\newcommand{\defined}[1]{#1\neq\bot}

\newcommand{\eqdef}{\stackrel{\mathrm{Def}}{=}}
\newcommand{\eqmu}{\stackrel{\mu}{=}}
\newcommand{\eqnu}{\stackrel{\nu}{=}}
\newcommand{\eqrec}{\stackrel{\mathrm{rec}}{=}}
\newcommand{\eqc}[1]{\stackrel{#1}{=}}

\newcommand{\rt}[2]{#2|_{#1}}

\newcommand{\myifthenelse}[3]
    {\mathbf{if}\,#1\,\mathbf{then}\,#2\,\mathbf{else}\,#3}

\newcommand{\strictapp}[2]{#1{\downarrow}#2}

\newcommand{\ire}[2]{#1\,\mathbf{r}\,#2}

\newcommand{\IFP}{\mathrm{IFP}}  
\newcommand{\CFP}{\mathrm{CFP}}  
\newcommand{\TCF}{\mathrm{TCF}} 

\newcommand{\rec}{\mathbf{rec}}

\newcommand{\rea}{\mathbf{R}}
\newcommand{\reah}{\mathbf{H}}

\newcommand{\acc}{\mathbf{Acc}}

\newcommand{\multwo}{\mathbf{mult}_2}
\newcommand{\appr}{\mathbf{approx}}

\newcommand{\reals}{\iota}
\newcommand{\creals}{F}

\newcommand{\one}{\mathbf{1}}
\newcommand{\ftyp}[2]{#1 \Rightarrow #2}
\newcommand{\tfix}[2]{\mathbf{fix}\,#1\,.\,#2}

\newcommand{\mycomment}[1]{}
\newcommand{\caseof}[1]{\mathbf{case}\, #1\, \mathbf{of}\, } 
\newcommand{\Am}{\mathbf{A}}
\newcommand{\less}{\prec}
\newcommand{\pat}{\mathbf{Path}}
\newcommand{\BT}{\mathbf{BT}}

\newcommand{\nsfmat}{F_{\mathrm{ns}}^{n\times n}}
\newcommand{\fmat}{F_{\mathrm{ns}}^{n\times n}}
\newcommand{\mnsfmat}{F_{(\mathrm{ns})}^{n\times n}}

\newcommand{\In}{\mathbf{In}}

\begin{document}

\chapter[Concurrent Gaussian elimination]{Concurrent Gaussian elimination}

\author[U.~Berger and M.~Seisenberger]{Ulrich Berger and Monika Seisenberger}

\address{Department of Computer Science, Swansea University Bay Campus,\\ 
Fabian Way, Crymlyn Burrows, Swansea SA1 8EN, UK\\
\{u.berger,m.seisenberger\}@swansea.ac.uk}

\author[D.~Spreen]{Dieter Spreen}
\address{Department of Mathematics, University of Siegen, \\
57068 Siegen, Germany\\
spreen@math.uni-siegen.de}

\author[H.~Tsuiki]{Hideki Tsuiki}
\address{Graduate School of Human and Environmental Studies,
  Kyoto University,\\ 
Kyoto, Japan\\
tsuiki@i.h.kyoto-u.ac.jp}

\author[U.~Berger, M.~Seisenberger, D.~Spreen, H.~Tsuiki]{}

\begin{abstract}
%


Working in a semi-constructive logical system that supports the extraction of
concurrent programs, we extract a program inverting non-singular 
real valued matrices from a constructive proof based on Gaussian elimination. 
Concurrency is used for efficient pivoting, that is, 
for finding an entry that is apart from zero
in a non-null vector of real numbers.

\end{abstract}


\body




\section{Introduction}
\label{sec-introduction}

A salient feature of constructive 
mathematics\cite{BishopBridges86,Bridges99,BridgesReeves99},
which attracts mathematicians, logicians, philosophers,
and computer scientists, independently of their foundational position,
is the fact that its notions are computationally
meaningful and its proofs give rise to concrete realizations of computations, 
that is, programs.
Manifestations of this ability to `extract programs'
can be found, for example, in constructive type theories, 
functional interpretation and theories of realizability. 
Using computer implementations of constructive type theory, such as 
Nuprl~\cite{Constable86}, 
Coq~\cite{CoqProofAssistant} and 
Agda\cite{Agda}, 
or systems supporting realizability, such as 
Isabelle/HOL~\cite{Berghofer03},
and 
Minlog~\cite{BergerMiyamotoSchwichtenbergSeisenberger11}, 
substantial case studies
have been undertaken that demonstrate the usefulness and scalability of
program extraction.

The exploration of the scope and limits of program extraction is an active 
field of research that has led to many interesting results regarding
the partial inclusion of 
classical logic~\cite{Constable91,Parigot92,Berger02} 
and mathematical principles such as various forms of 
choice~\cite{Berardi98,BergerOliva01,Krivine03,Berger04,Seisenberger08} 
and induction
principles~\cite{Berger04,Schuster13,BergerTsuiki21,PowellSchusterWiesnet21}.
The benefit of program extraction is not only that extracted programs come
automatically with formal correctness proofs but also that important 
additional properties of programs concerning computational complexity,
totality and typability can be obtained. 

An important paradigm in modern computing technology, that has
only recently been connected with program extraction, is \emph{concurrency},
that is, the ability to run several `threads' of computations independently
and asynchronously but allowing them to interact in a number of ways.
Concurrency poses big challenges to computer science, since,
on the one hand, it is indispensable for running computer systems efficiently
on a range of different tasks simultaneously, on the other hand the 
(often nondeterministic) behaviour of such systems is very difficult to
control.

In\citep{BergerTsuikiCFP}, a semi-constructive logical system 
($\CFP$, Concurrent Fixed Point Logic) for the extraction of concurrent programs
was introduced and applied to extract a program that converts between two 
different exact representations of real numbers. Concurrency was needed to
avoid non-termination when dealing with potentially partial inputs 
(infinite Gray code~\citep{Tsuiki02}). In the present paper, we apply 
$\CFP$ to another problem
in computable analysis, the inversion of real valued matrices. Here, not
non-termination is the issue, but the potential gain in efficiency through
concurrent computation.  
In both applications the interaction between concurrent threads 
is limited to winning a race: The result of the thread terminating first 
will be selected and the other killed. Nevertheless, controlling
this form of concurrency at the logical level poses serious challenges.

The paper is organized as follows:
In Sec.~\ref{sec-ifp} we introduce the main elements of the 
system~$\IFP$~\citep{BergerTsuiki21} our method of program extraction is based on, 
and in Sec.~\ref{sec-realnumbers} we instantiate $\IFP$ to the setting of 
real numbers. 
In Sec.~\ref{sec-concurrency} we present the system $\CFP$ that extends $\IFP$
by logical operators permitting the extraction of concurrent programs.
To be able to carry out our case study, we  adjust $\CFP$ to allow for
more than two concurrent threads and iterated `monadic' concurrency.
Finally, Sec.~\ref{sec-gauss} presents the proof of matrix inversion through
Gaussian elimination where pivoting, that is selecting a nonzero entry in a 
nonzero vector, is done concurrently.


\section{Intuitionistic fixed point logic}
\label{sec-ifp}
Our basic logical system for program extraction is $\IFP$,
which is based on many-sorted intuitionistic first-order logic with equality and
least and greatest fixed points. 
A detailed definition of its syntax, semantics and proof calculus,
as well as applications to program extraction can be found 
in~\citep{BergerTsuiki21}.
We highlight the most important points.

The \emph{syntax} of $\IFP$ is parametric in a set of sorts 
(which we usually suppress), a set of constants and function symbols, 
and a set of predicate constants (all with fixed arities).
\emph{Formulas} $A,B$, \emph{predicates} $P,Q$, 
and \emph{operators} $\Phi,\Psi$ are defined simultaneously:
%

%
\emph{Formulas} 
are $\bot, \top$ (falsity and truth), $s=t$ where $s,t$ are first-order 
terms of the same sort, $P(\vec t)$ where $\vec t$ is a tuple of first-order terms 
whose length is the arity of the predicate $P$, $A\land B$, $A\lor B$, $A\to B$, and
$\forall x\,A$, $\exists x\,A$ where $x$ is a first-order variable of a fixed sort.

\emph{Predicates} 
are predicate variables, predicate constants, both of fixed arities, 
$\lambda \vec x\,A$ (abstraction) where $\vec x$ is a tuple of first-order variables
and $\mu\,\Phi$, $\nu\,\Phi$. The common arity of the latter two constructs
is the arity of the operator $\Phi$.

\emph{Operators} 
are of the form $\lambda X\,P$ where $X$ is a predicate 
variable of the same arity as $P$ which occurs freely in $P$ only at strictly 
positive positions, that is, not in the left part of an implication.
The arity of $\lambda X\,P$ is the common arity of $X$ and $P$.
%

%
The application, $\Phi(Q)$, of an operator $\Phi = \lambda X\,P$ to a predicate
$Q$ of the same arity is defined as $P[Q/X]$.
A definition $P \eqdef \mu \Phi$ will also be written as $P \eqmu \Phi(P)$. 
The notation $P \eqnu \Phi(P)$ has a similar meaning.
If $\Phi = \lambda X \lambda \vec x\,A$, then we also write
$P(\vec x) \eqmu A[P/X]$ and $P(\vec x) \eqnu A[P/X]$ 
instead of $P \eqdef \mu \Phi$ and $P \eqdef \nu \Phi$.
We furthermore use the set-theoretic notations 
$P \subseteq Q$ for $\forall x\,(P(x)\to Q(x))$ 
and $P \cap Q$ for $\lambda x\,(P(x)\land Q(x))$, 
as well as the bounded quantifiers $\forall x\in P\,A$ and $\exists x\in P\,A$
for $\forall x\,(P(x)\to A)$ and $\exists x\,(P(x)\land A)$.
In Sects.~\ref{sec-concurrency} and~\ref{sec-gauss} we will also use finitely
iterated conjunctions and disjunctions which we define as
\begin{eqnarray*}
\bigvee_{1\le i\le n}A(i) 
     &\eqdef&   A_1\land (A_2\land \ldots (A_n\land\top)\ldots)\\
\bigwedge_{1\le i\le n}A(i) 
    &\eqdef&   A_1\lor (A_2\lor \ldots (A_n\lor\bot)\ldots).
\end{eqnarray*}

The \emph{semantics} of $\IFP$\index{semantics of $\IFP$} has a classical and a 
constructive part.

The \emph{classical semantics} is as usual for classical first-order predicate logic,
that is, sorts are interpreted as sets (called carrier sets), 
function symbols as (set-theoretic) functions, and predicate constants 
as subsets of the corresponding cartesian products of carrier sets.
Operators are interpreted as monotone predicate transformers and
the predicates $\mu\,\Phi$ and $\nu\,\Phi$ as their least and greatest fixed point,
also called inductive and coinductive definitions.
No assumptions of computability or constructivity are made at this point.

The \emph{constructive semantics} consist of a realizability 
interpretation\index{realizability interpretation} of
formulas and (co)inductive predicates. The space of potential realizers
is an effective Scott-domain\index{Scott-domain}
\citep{GierzHofmannKeimelLawsonMisloveScott03}, 
$D$, defined by the recursive domain equation
\begin{equation}
\label{eq-D}
D = (\Nil + \Left(D) + \Right(D) + \Pair(D\times D) 
  + \Fun(D\to D))_\bot 
\end{equation}
This means that $D$ is the disjoint union of the components labelled
by the constructors $\Nil, \Left,\ldots$. $D$ contains
its own function space, $(D \to D)$, which would not be possible, 
set-theoretically, but is a standard construction in domain theory, since
$(D \to D)$ consists of continuous functions only. Continuity refers to
the Scott topology which is generated by a partial order $\sqsubseteq$ on $D$
which has $\bot$ as its least element (representing `undefined').
Thanks to the function space component, $D$ carries the structure of a PCA 
(partial combinatory algebra)\index{partial combinatory algebra} and can 
therefore interpret simple recursive types defined by the grammar
\begin{equation}
\label{eq-types}
\rho, \sigma ::=  \alpha \mid \one  
                         \mid \rho \times \sigma
                         \mid \rho + \sigma
                         \mid \ftyp{\rho}{\sigma}
                         \mid \tfix{\alpha}{\rho}
\end{equation}
where $\alpha$ ranges over type variables and in $\tfix{\alpha}{\rho}$,
$\rho$ must be strictly positive in $\alpha$. Hence, every type $\rho$
denotes a subdomain $D(\rho)$.
To define the semantics of a formula $A$ we first define a type
$\tau(A)$\index{type of a formula}
that represents the `propositional skeleton' of $A$.
Then the semantics of $A$ is defined as a set $\rea(A) \subseteq D(\tau(A))$ 
of realizers of $A$, essentially following Kleene~\citep{Kleene45} and 
Kreisel~\citep{Kreisel59} but, if possible, simplifying realizers in order to avoid
junk in extracted programs. 
We will usually write `$\ire{a}{A}$' (`$a$ realizes $A$') 
for `$a\in\rea(A)$' and omit the type information about $a$.
For the simplification of realizers, \emph{Harrop formulas} play a special role.
These are formulas that do not contain disjunctions or free 
predicate variables at strictly positive positions.
For example, $\bot$, equations and formulas of the form every 
$P(\vec t)$ where $P$ is a predicate constant are Harrop,
and so is every negated formula since we view negation as 
implication to falsity.
The realizability interpretation of an implication $A\to B$  where
$A$ and $B$ are both non-Harrop is, as expected,
\[ \ire{c}{(A\to B)} \eqdef 
   \forall a\,(\ire{a}{A} \to\ire{(c\,a)}{B}) \] 
while if $A$ is Harrop it is
\[ \ire{b}{(A\to B)} \eqdef 
      \reah(A) \to\ire{b}{B}\,. \] 
Here $\reah(A)$ means $\exists a\,(\ire{a}{A})$ or, equivalently
$\ire{\Nil}{A}$ since, for Harrop formulas, $\Nil$ (an element of $D$ signalling 
trivial computational content) is the only possible realizer.
An even more restricted class of formulas are \emph{non-computational (nc)} formulas 
which do not contain free predicate variables or disjunctions at all. 
These formulas coincide with their realizability interpretation, that is,
$\reah(A)$ is identical to $A$ for nc formulas $A$.
Since first-order variables may range over abstract objects without computable 
structure, quantifiers have to be interpreted uniformly, that is,
\begin{eqnarray*} 
\ire{a}{\forall x\,A(x)} &\eqdef& \forall x\,(\ire{a}{A(x)})\\
\ire{a}{\exists x\,A(x)} &\eqdef& \exists x\,(\ire{a}{A(x)})\,.
\end{eqnarray*}
Hence, the realizer of a universal formula does not depend on the quantified variable,
and the realizer of an existential formula does not carry a witness. 

The realizability interpretation of a least or greatest fixed point is the 
least or greatest fixed point of the realizability interpretation of the corresponding
operator. For example, in the presence of a constant $0$ and a successor function
$(+1)$ we can define the natural numbers  as the least predicate containing $0$
and being closed under $(+1)$,
\begin{equation}
\label{eq-nat}
\NN(x) \eqmu x = 0 \lor \exists\, y\, (x=y+1 \land\NN(y)) 
\end{equation}
The realizability interpretation of $\NN$ is
\begin{equation}
\begin{array}{l}
\label{eq-natr}
\ire{n}{\NN(x)}\ \eqmu\ 
        (n=\Left(\Nil)\land x = 0)\ \lor\ \\ 
\hspace*{5.7em}    
        (n=\Right(m)\land\exists y\,(\ire{m}{\NN(y)} \land x = y+1)) 
\end{array}
\end{equation}
Therefore, $\ire{n}{\NN(x)}$ means that $x$ is a natural number and 
$n$ is the unary representation of $x$ where zero and successor are
modelled by $\Left(\Nil)$ and $\Right(\_)$. 
Examples of coinductive definitions will be given in Sect.~\ref{sec-realnumbers}.

The guiding principle for the \emph{proof rules} of $\IFP$ is 
that they must be sound for the classical and the constructive semantics.
This is achieved through the usual rules for intuitionistic logic, and
rules expressing the least resp.~greatest fixed point property of
$\mu\,\Phi$ resp.~$\nu\,\Phi$. For the natural numbers these are equivalent to
the closure under $0$ and successor and the usual induction rule 
that infers, for an arbitrary predicate $P$, 
from $P(0)$ and $\forall x\,(P(x)\to P(x+1))$ that $\NN\subseteq P$, that is,
$\forall x\,(\NN(x)\to P(x))$.

Regarding \emph{axioms} one has a fair amount of flexibility. Any closed formula $A$
that is true in the intended classical interpretation and for which one has a 
program $M$ that realizes it can be taken as an axiom. 
In particular every classically true nc formula can be taken as an axiom. 
More generally, any closed Harrop formula that is equivalent to its realizability 
interpretation, may be taken as an axiom if we accept it as true in the intended 
interpretation. A large class of formulas enjoying this property is the class 
of \emph{admissible Harrop formulas}, which is the least class that contains all 
nc formulas, is closed under conjunction and quantification and under 
implication, $A\to B$ where $B$ is admissibly Harrop and the formula $A$ is equivalent 
to its realizability interpretation (for example, that is the case 
if $A$ is of the form $\NN(t)$).
An important and useful axiom is a generalization of 
\emph{Brouwer's thesis}~\citep{Veldman06} which
can be stated as the following nc formula
\begin{equation}
\label{eq-bt}
\BT\qquad\forall x\,(\neg\pat_{\less}(x) \to \acc_{\less}(x)) \,.
\end{equation}
This refers to some nc relation $\less$ and relates two ways of expressing 
that an element $x$ is wellfounded with respect to $\less$:
The first way, $\neg\pat_{\less}(x)$, says that there is no descending 
path starting with $x$ 
where the existence of a path can be expressed coinductively as
\begin{equation}
\label{eq-path}
\pat_{\less}(x) \eqnu \exists y \less x\, \pat_{\less}(y)\,.
\end{equation}
The second way, $\acc_{\less}(x)$, expresses wellfoundedness inductively
as the accessible part of $\less$:
\begin{equation}
\label{eq-path1}
\acc_{\less}(x) \eqmu \forall y \less x\, \acc_{\less}(y)
\end{equation}
$\BT$ is stronger and more general than Brouwer's original thesis. The precise
relationship is explained in~\citep{BergerTsuiki21}.
As an example we show:
\begin{lemma}
\label{lem-mp}
Brouwer's thesis, (\ref{eq-bt}), implies Markov's Principle, that is, the schema
\begin{equation}
\label{eq-mp}
\forall x\in\NN (A(x)\lor \neg A(x)) \land \neg\neg\exists x\in\NN\,A(x) \to \exists x\in\NN\,A(x)\,.
\end{equation}
\end{lemma}
\begin{proof}
To prove~(\ref{eq-mp}) set $y \less x \eqdef y = x+1 \land \neg A(x)$. 
Then, clearly, $\neg\neg\exists x\in\NN\,A(x)$ implies $\neg\pat_{\less}(0)$
(if $\pat_{\less}(0)$, then $\forall x\in\NN\,\pat_{\less}(x)$ follows by induction on $x\in\NN$, hence $\forall x\in\NN\neg A(x)$).
Therefore, by (\ref{eq-bt}), $\acc_{\less}(0)$. 
We set $P(x) \eqdef x\in\NN \to \exists y\in\NN\,A(y)$, 
with the goal to prove $\acc_{\less}\subseteq P$, by wellfounded induction 
on $\less$, that is,
strictly positive induction on $\acc_{\less}$. 
Therefore, we assume, as induction hypothesis,  $\forall y \less x\, P(y)$, and
have to show $P(x)$. Hence let $x\in\NN$. Since $A$ is assumed to be decidable, 
we can do case analysis on $A(x)$. If $A(x)$ holds, then we are done. 
Otherwise, $x+1 \less x$ and therefore $P(x+1)$ holds by the induction 
hypothesis. Since $x+1\in\NN$, it follows $\exists y\in\NN\, A(y)$, 
thus completing the inductive proof of $\acc_{\less}\subseteq P$.
Hence $P(0)$ which means $\exists x\in\NN\,A(x)$.
\end{proof}
Since Markov's Principle, and hence $\BT$, is not generally accepted in 
constructive mathematics, we will mark uses of $\BT$ in the following.

\emph{Program extraction} is based on the soundness theorem for the 
constructive semantics:
\begin{theorem}[Soundness]
\label{thm-soundness}
%
From an $\IFP$ proof of a formula $A$ one can extract a closed program 
$M:\tau(A)$ such that $\ire{M}{A}$ is provable.
\end{theorem}
A detailed proof can be found in~\citep{BergerTsuiki21}.
%

The realizers extracted from proofs are given as programs 
in a simple functional programming language, that is,
an enriched lambda calculus given by the grammar
\begin{align*}
&  M,N,L :: = 
a,b,c \ \ \text{(variables)}\\
&\quad  |\  \lambda a.\,M 
\ | \ M\,N 
\ | \ \rec\,M \ | \  \bot \\
&\quad |\    \Nil\ | \ \Left(M)\ | \ \Right(M)\ | \ \Pair(M,N)\\
&\quad  |\
 \caseof{M} \{\ Cl_1,\ldots, Cl_n\} 
%
\end{align*}
The $Cl_i$ are clauses of the form $C(a_1,\ldots,a_k) \to N$
where $C$ is a constructor  ($\in \{\Nil,\Left,\Right,\Pair\}$).
A program $\rec\,M$ denotes the least fixed point of the function defined by $M$.
Details of the denotational semantics of programs can be found 
in~\citep{BergerTsuiki21}. There, also a matching operational semantics is given.

A system similar to $\IFP$ is $\TCF$ 
(Theory of Computable Functionals)\index{Theory of Computable Functionals}\index{$\TCF$},
which is the formal system underlying the interactive proof system 
Minlog\index{Minlog}~\citep{SchwichtenbergWainer12,BergerMiyamotoSchwichtenbergSeisenberger11}. A comprehensive account of $\TCF$ and examples of program extraction
in Minlog can be found in~\citep{SchwichtenbergWainer12}. 
The main difference between $\IFP$ and $\TCF$ is that while $\IFP$ concentrates
on the integration of abstract and constructive mathematics, the focus in $\TCF$
is on higher type computable functionals as a standard model.

In the following, if we say of a formula that it ``holds'' or ``implies'' 
another formula we always mean that the corresponding formulas are provable 
in the formal system in question, which will be $\IFP$ in the next two sections,
and its concurrent variant, $\CFP$, thereafter.


\section{Real numbers}
\label{sec-realnumbers}
We consider an instance of $\IFP$ with a sort $\reals$ 
for real numbers, the constants $0$ and $1$ and 
function symbols for addition subtraction, multiplication, division, and
absolute value, as well as relation constants $<$ and $\le$.
Furthermore, we assume the axioms of a real closed field as a set of nc formulas,
as well as the \emph{Archimedean axiom} that says that there are no infinite 
numbers, that is, by repeatedly subtracting $1$ any real number will eventually 
become negative. Setting $y \less x \ \eqdef \ 0 \le y \ \land\  y = x-1$, 
this can be expressed as the nc formula
\begin{equation}
\label{eq-ap}
\forall x\,\neg\pat_{\less}(x)\,.
\end{equation}

Hence, real numbers are given abstractly and axiomatically. 
In particular, no programs for computing the field operations are available.
However, we can use~(\ref{eq-nat}) to define the natural numbers as a subset 
$\NN$ of the reals and prove, for example, that they are closed
under addition,
\begin{equation}
\label{eq-natadd}
\forall x, y\,(\NN(x) \land \NN(y) \to \NN(x+y))\,.
\end{equation}
The extracted program will add (unary representations of) natural numbers. 
To extract programs for (exact) real number computation, we need to be able
to approximate real numbers by rational numbers. Therefore we define the predicate
\begin{equation}
\label{eq-approx}
\A(x) \ \eqdef \ \forall n \in\NN\, \exists q \in\QQ\, |x-q|\le 2^{-n}
\end{equation}
where $\QQ$ is a predicate characterising the rational numbers, for example,
\begin{equation}
\label{eq-rat}
\QQ(x) \ \eqdef \ \exists y,z,u\in\NN\ x = \frac{y-z}{u+1}\,.
\end{equation}
We call a real number $x$ satisfying $\A(x)$ \emph{constructive}.
A realizer of $\A(x)$
is a fast Cauchy sequence, that is, a 
sequence $f$ of rational numbers converging quickly to $x$:
\begin{equation}
\label{eq-rat1}
\ire{f}{\A(x)} \ \leftrightarrow \ \forall n\in\NN\,|x-f(n)|\le 2^{-n}\,.
\end{equation}
The realizers of $\A(x)$ are sequences of rational numbers implemented as
\emph{functions} from the natural numbers to the rationals.
To obtain as realizers sequences that are implemented as infinite \emph{streams}, 
one can use the following coinductive predicate:
\begin{equation}
\label{eq-approxco}
\Aco(x) \eqdef \Aco(x,0)\hbox {, where }\ 
\Aco(x,n) \eqnu \exists q\in\QQ |x-q|\le 2^{-n} \land \Aco(x,n+1)
\end{equation}
Then a realizer of $\Aco(x)$ is an infinite stream of rationals $q_0:q_1:\ldots$
such that $|x-q_n| \le 2^{-n}$ for all $n$.
It is easy to prove that the predicates $\A$ and $\Aco$ are equivalent.
From the proof of this equivalence one can extract programs that translate 
between the function and the stream representation of Cauchy sequences.

In~\citep{BergerTsuiki21} it is shown that $\A$ is also equivalent to the
coinductive predicate
\begin{equation}
\label{eq-sd}
\C(x) \eqnu \exists d \in \{-1,0,1\}\, (|d/2-x|\le 1/2 \wedge \C(2x-d))\,.
\end{equation}
A realizer of $\C(x)$ is a signed digit representations of $x$.
The programs extracted from the proof of the equivalence of $\A$ and $\C$
provide translations between the Cauchy and the signed digit representations.
The coinductive characterization of the signed digit representation has
been generalized to continuous functions~\citep{Berger11}
and compact sets~\citep{BergerSpreen16} and corresponding translations
have been extracted.

A predicate of arity $(\iota)$ is called a \emph{ring} if it
contains $0$ and $1$ and is closed under addition, subtraction and 
multiplication. A ring that is closed under division by elements that
are apart from $0$ is called a \emph{field}, where the 
\emph{apartness relation} between real numbers is defined as
\begin{equation}
\label{eq-apart}
x\,\apart\,y \ \eqdef \ \exists k \in\NN\,|x-y|\ge 2^{-k}
\end{equation}
Note that, in our setting, rings and fields are automatically ordered and
satisfy the Archimedean axiom~(\ref{eq-ap}).
A \emph{constructive field} is a field $F$ whose elements are all constructive,
that is, $F \subseteq \A$.
\begin{lemma}
\label{lem-cfield}
$\QQ$, $\A$, $\Aco$, and $\C$ are constructive fields.
\end{lemma}
Recall that statements about validity of formulas are to be understood 
as provability in $\IFP$. Therefore, a constructive field provides,
via program extraction, implementations of the field operations with respect
to the real number representation defined by the field. 
Since closure under the multiplicative inverse is required only for
numbers that are apart from $0$, the extracted inversion program has
as extra argument a natural number $k$ that guarantees that $|x|\ge
2^{-k}$. In the following, we show that this extra argument can be
omitted if one is willing to accept the generalized Brouwer's Thesis,
$\BT$.

Using $\BT$, the Archimedean axiom (\ref{eq-ap}) becomes 
$\forall x\,\acc_{\less}(x)$ (where, as in (\ref{eq-ap}), 
$y \less x \ \eqdef \ 0 \le y \ \land\  y = x-1$), 
and this can in turn be used to prove 
the following two induction rules, called 
\emph{Archimedean Induction}~\citep{BergerTsuiki21}.
Both rules are parametrized by a positive rational number $q$.
In~(\ref{eq-aiprimitive}), $P$ is an arbitrary predicate.
In~(\ref{eq-ai}), $R$ is a ring and $B$ is arbitrary. 
\begin{equation}
\label{eq-aiprimitive}
\infer[]{
  \forall x \ne 0\,  P(x)
}{
 \forall x \ne 0\, ((|x| \leq q \to P(2x)) \to P(x))
}
\end{equation}
\begin{equation}
\label{eq-ai}
\infer[]{
  \forall x \in R\setminus\{0\} \,B(x) 
}{
\forall x\in R\setminus\{0\}\,(B(x)\lor(|x| \leq q \land (B(2x) \to B(x))))
}
\end{equation}
Unlike~(\ref{eq-ap}), the Archimedean Induction rules do have computational content:
\begin{lemma}[$\BT$, \citep{BergerTsuiki21}]
\label{lem-ai}
(\ref{eq-aiprimitive}) and (\ref{eq-ai}) are provable.

If $s$ realizes the premise of (\ref{eq-aiprimitive}), then the extracted program 
realizing the conclusion is $\rec\,s$.

If $s$ realizes the premise of~(\ref{eq-ai}), then the extracted program 
realizing the conclusion is
\[\chi\,a \eqrec \mycase\,s\,a\,\of\,\{\Left(b) \to b; 
                                \Right(f) \to f(\chi\,(\multwo(a)))\}\]
where $\multwo$ is extracted from a proof that $R$ is closed under doubling.
\end{lemma}
\begin{proof}
(\ref{eq-aiprimitive}) can be inferred from $\forall x\,\acc_{\less}(x)$ 
using properties of the function $x \mapsto 2^{-x}$.

(\ref{eq-ai}) follows directly from (\ref{eq-aiprimitive}),
since with $P(x) \eqdef R(x) \to B(x)$, the premise of (\ref{eq-ai}) implies
the premise of (\ref{eq-aiprimitive}).
\end{proof}

\begin{lemma}[$\BT$]
\label{lem-apart}
Non-zero elements of a constructive fields are apart from zero.
\end{lemma}
\begin{proof}
Let $F$ be a constructive field. We have to show 
$\forall x\in F\setminus\{0\}\,x\apart 0$.
We use Archmedean Induction (\ref{eq-ai}) with $q=3$ and the
fact that $2x\apart 0 \to x\apart 0$ to
reduce this task to showing that for every $x\in F\setminus\{0\}$, 
$x\apart 0$ or $|x| \leq 3$.
Using the constructivity of $F$, we have $q\in\QQ$ such that 
$|x-q|\le 1$. If $|q|>2$, then $|x| \ge 1$, 
hence $x\apart 0$. If $|q| \le 2$, then $|x|\le 3$.
\end{proof}
Since the implication $2x\, \apart\,0 \to x\, \apart\,0$ is realized by 
the successor function, the program extracted from the above proof is
(after trivial simplifications)
\begin{eqnarray*}
\varphi_1(f) &=& 
  \myifthenelse{|\appr(f,0)| > 2}{0}{1 + \varphi_1 (\multwo(f))}
\end{eqnarray*} 
where $\appr$ realizes the constructivity of $F$, that is,
if $f$ realizes $F(x)$ and $k\in\NN$, 
then $\appr(f,k)$ is a rational number whose
distance to $x$ is at most $2^{-k}$.

\emph{Remark.} It may seem strange that we introduce constructive analysis
in an axiomatic way; however, this approach is not uncommon in
constructive mathematics. For example, \citep{BridgesReeves99} introduces
constructive axioms for the real line, 
\citep{BridgesVita03,BridgesIshiharaSchusterVita08} study an axiomatic
theory of apartness and nearness, and~\citep{Petrakis16} gives a constructive
axiomatic presentation of limit spaces.
In our case, the rational for an axiomatic approach to the real numbers is that
it allows us to separate computationally irrelevant properties of the real
numbers such as algebraic identities and the Archimedean axiom from
computationally relevant ones such as inclusion of the integers and
approximability by rationals.


\section{Concurrency}
\label{sec-concurrency}
In~\citep{BergerTsuikiCFP}, the system $\CFP$ (concurrent fixed point logic)
is introduced that extends $\IFP$ by a logical operator $\Set$ permitting 
the extraction of programs with two concurrent threads. However, 
to do Gaussian elimination concurrently, we need $n$ threads where $n$
varies through all positive numbers less or equal the dimension of the matrix 
to be inverted.
We will first sketch the essential features of $\CFP$ and then
describe an extension that caters for this more general form of concurrency. 

\subsection{$\CFP$}
\label{sec-cfp}
%
%
%
$\CFP$ extends $\IFP$ by two logical operators, $\rt{B}{A}$ (restriction) and
$\Set(A)$ (two-threaded concurrency). 
Ultimately, we are interested 
only in the latter but the former is needed to derive
concurrency in a non-trivial way.  

$\rt{B}{A}$ (read `$A$ restricted to $B$') 
is a strengthening of the reverse implication $B \to A$.
The two differ only in their realizability semantics
(the following explanation is valid if $B$ is non-computational
and $A$ has only realizers that are defined ($\neq\bot$)): 
A realizer of  $\rt{B}{A}$ is a `conditional partial realizer' of $A$, that is, 
some $a\in \tau(A)$ such that 
\begin{itemize}
\item[(i)] if the restriction $B$ holds, then $a$ is defined, and 
\item[(ii)] if $a$ is defined, then $a$ realizes $A$. 
\end{itemize}
In contrast, for $a$ to realize $B\to A$ it suffices that if $B$ holds, then 
$a$ realizes $A$; condition (ii) may fail, that is, if we only know that 
$a$ is defined, we cannot be sure that $a$ realizes $A$.
For example, if $A$ is $B\lor C$, then $\Left(\Nil)$ realizes 
$B \to A$ but not $\rt{B}{A}$ unless $B$ holds. 
The technical definition permits arbitrary $B$ but requires $A$ to satisfy a
syntactic condition, called 
`strictness', 
that ensures that $A$ has 
only defined realizers.

Before defining strictness we need to extend the Harrop property to $\CFP$ formulas.
Harropness is defined as for $\IFP$ but stipulating in addition that formulas
of the form $\rt{B}{A}$ or $\Set(A)$ are always non-Harrop.
Now strictness is defined as follows:
A non-Harrop implication is strict if the premise is non-Harrop.
Formulas of the form $\diamond x\,A$ ($\diamond\in\{\forall,\exists\}$) or
$\munu(\lambda X\lambda\vec x\,A)$ ($\munu\in\{\mu,\nu\}$) are
 \emph{strict} if $A$ is strict.
Formulas of other forms (e.g., $\rt{A}{B}$,  $\Set(A)$, $X(\vec{t})$) 
are not strict.

For formulas of the for $\rt{B}{A}$ or  $\Set(A)$ to be wellformed, 
we require $A$ to be strict. In the following it is always assumed 
that the strictness requirement is fulfilled.

Realizability for restriction is defined as
\begin{eqnarray*} 
a\,\re\,(\rt{B}{A}) &\eqdef& (\re(B) \to \defined{a}) \land
                               (\defined{a} \to a\,\re\,A)
\end{eqnarray*}
where $\re(B) \eqdef \exists b\,b\,\re\,B$ (`$B$ is realizable').

$\Set(A)$ (read `concurrently $A$') is not distinguished from $A$ in the classical
semantics but realizers of this formula can be computed by two processes which 
run concurrently, at least one of which has to terminate, and each 
of terminating one has to deliver a realizer of $A$.
To model this denotationally, the domain 
and the programming language are 
extended by a binary constructor
$\Amb$ which (denotationally) is an exact copy of the constructor $\Pair$. 
However, operationally, $\Amb$ is interpreted as a version of McCarthy's 
ambiguity operator amb~\citep{McCarthy1963} that corresponds to 
globally angelic choice~\citep{ClingerHalpern85}. 
%
Realizability of $\Set(A)$ is defined as
\begin{multline*}
c\ \re\ \Set(A)
  \eqdef 
   c = \Amb(a, b) \land a,b:\tau(A) \land 
           (\defined{a} \lor \defined{b})\ \land\\
 (\defined{a} \to a\, \re\, A) \land 
            (\defined{b} \to b\, \re\, A).
\end{multline*}
We also set $\tau(\Set(A)) \eqdef\Am(\tau(A))$ where $\Am$ is a new 
type constructor with 
$D(\Am(\rho))\eqdef\{\Amb(a,b) \mid a,b\in D(\rho)\}\cup\{\bot\}$.

The proof rules for restriction and concurrency are as follows:

\[
\infer[\hbox{Rest-intro ($A_0, A_1, B$ Harrop)}]{
        \rt{B}{(A_0 \vee A_1)} 
}{
B \to A_0 \vee A_1 \ \ \     \neg B \to A_0 \wedge A_1
}
\]

\[
\begin{array}{ll}
\infer[\hbox{Rest-bind}]{
      \rt{B}{A'}
}{
 \rt{B}{A}\ \ \          A \to (\rt{B}{A'})
}
\ \ \ \ \ \ \ \ & 
\infer[\hbox{Rest-return \ \ \ }]{  
 \rt{B}{A}
}{
  A
}  \\\\
  \infer[\hbox{Rest-antimon}]{
    \rt{B'}{A}
    }{
    \rt{B}{A} \ \ \    B' \to B
}&
  \infer[\hbox{Rest-mp}]{
    A
}{
\rt{B}{A} \ \ \    B
}
\end{array}
\]
\[
\begin{array}{ll}
  \infer[\hbox{Rest-efq}]{
  \rt{\False}{A}
}{
}
\ \ \ \ \ \ \ \ &
\infer[\hbox{Rest-stab}]{
   \rt{\neg\neg B}{A} 
    }{
    \rt{B}{A}
}
\end{array}
\]

\[
  \infer[\hbox{Conc-lem}]{
  \Set(A)
}{
\rt{C}{A}  \ \ \ \     \rt{\neg C}{A}
}
\qquad
  \infer[\hbox{Conc-return}]{\ 
  \Set(A)
}{
A
}
\]

\[
  \infer[\hbox{Conc-mp}]{\
\Set(B)
}{
  A\to B\ \ \  \Set(A) 
}\ \ \ \ \ 
%
\]
Using the rules (Rest-stab) and (Rest-antimon) on can show that the rule 
(Conc-lem) is equivalent to
\begin{equation}
\label{eq-concclassor}
  \infer[\hbox{Conc-class-or}]{
  \Set(A)
}{
 \neg\neg(B \lor C)  \ \ \ \   \rt{B}{A}   \ \ \ \   \rt{C}{A}
}
\end{equation}

We also have a variant of the Archimedean Induction which we call
\emph{Archimedean Induction with restriction}. It applies to 
rings $R$, 
strict 
predicates $B$  and rational numbers $q>0$:
\begin{equation}
\label{eq-airest}
 \infer[]{
   \forall x \in R\, (\rt{x\neq 0}{B(x)})
 }{
  \forall x \in R\, 
      (B(x) \lor (|x| \le q \land (B(2x) \to B(x))))
 }
\end{equation}

To extract programs from $\CFP$ proofs, the programming language for $\IFP$
needs, in addition to the new constructor $\Amb$, a strict version of 
application $\strictapp{M}{N}$ which is undefined if the argument $N$ is.
As indicated earlier, denotationally, $\Amb$
is just a constructor (like $\Pair$). However, the operational 
semantics interprets 
$\Amb$ as globally angelic choice, which matches the realizability 
interpretation of the concurrency operator $\Set$, 
as shown in~\citep{BergerTsuikiCFP}.
\begin{theorem}[Soundness of CFP]
\label{thm-soundnesscfp}
%
From a $\CFP$ proof of a formula $A$ one can extract a closed 
program 
$M:\tau(A)$ such that $\ire{M}{A}$ is provable.
\end{theorem}
\begin{proof}
The realizability of the new version (\ref{eq-airest}) of 
Archimedean Induction is shown in the Lemma below.
All other rules have simple 
realizers that can be explicitly defined~\citep{BergerTsuikiCFP}.
For example, (Conc-lem) and (Conc-class-or), rules which permit 
to derive concurrency with a form of the law of excluded middle, are 
realized by the constructor $\Amb$. 
\end{proof}

\begin{lemma}
\label{lem-air}
Archimedean Induction with Restriction is realizable.
If $s$ realizes the premise of (\ref{eq-airest}),
then $\chi$, defined recursively by
\[\chi\,a = \mycase\,s\,a\,\of\,\{\Left(b) \to b; 
                             \Right(f) \to \strictapp{f}{(\chi\,(\multwo\,a))}\},\]
realizes the conclusion. 
%
%
\end{lemma}
\begin{proof}
Assuming $\ire{a}{R(x)}$ we have to show
\begin{itemize}
\item[(1)] $x\ne 0 \to \defined{\chi\,a}$.
\item[(2)] $\defined{\chi\,a} \to \ire{(\chi\,a)}{B(x)}$.
\end{itemize} %
For (1) it suffices to show (1') 
$\forall x\neq 0 \,\forall a\,(\ire{a}{R(x)} \to \ire{(\chi\,a)}{B(x)})$ 
since $\ire{(\chi\,a)}{B(x)}$ implies $\defined{\chi\,a}$, 
by the strictness of $B(x)$.

We show (1') by Archimedean Induction~(\ref{eq-aiprimitive}). 
Let $x\neq 0$ and assume, as i.h.,
\[|x| \le q \to \forall a'\,(\ire{a'}{R(2x)} \to \ire{(\chi\,a')}{B(2x)})\]  
We have to show $\forall a\,(\ire{a}{R(x)} \to \ire{(\chi\,a)}{B(x)})$.
Assume $\ire{a}{R(x)}$. Then
$\ire{(s\,a)}{(B(x) \lor (|x| \le q \land (B(2x) \to B(x))))}$.
If $s\,a = \Left(b)$ where $\ire{b}{B(x)}$, then $\chi\,a = b$ and we are done.
If $s\,a = \Right(f)$ where $|x|\le q$ and
$\ire{f}{(B(2x) \to B(x))}$, then $\ire{a'}{R(2x)}$ for $a' = \multwo\,a$. 
Thus, by i.h., $\ire{(\chi\,a')}{B(2x)}$
and therefore $\ire{(f(\chi\,a'))}{B(x)}$. Since  $\defined{\chi\,a'}$
by the strictness of $B(2x)$, we have $\chi\,a = f(\chi\,a')$ and we are done.

To prove (2),
we consider the approximations of $\chi$,
\begin{eqnarray*}
\chi_0\,a &=& \bot\\
\chi_{n+1}\,a &=& \mycase\,s\,a\,\of\,\{\Left(b) \to b; 
                             \Right(f) \to \strictapp{f}{(\chi_n\,(\multwo\,a))}\}
\end{eqnarray*}
By continuity, if $\defined{\chi\,a}$, then $\defined{\chi_n\,a}$ for some
$n\in\NN$. Therefore, it suffices to show
\[ \forall n\in\NN\,\forall x,a\,(\ire{a}{R(x)} \land \defined{\chi_n\,a} 
     \to \ire{(\chi\,a)}{B(x)})\]
We induce on $n$. The induction base is trivial since $\chi_0\,a = \bot$.
For the step assume $\ire{a}{R(x)}$ and $\defined{\chi_{n+1}\,a}$. Then 
$\ire{(s\,a)}{(B(x) \lor (|x| \le q \land (B(2x) \to B(x))))}$.
If $s\,a = \Left(b)$ where $\ire{b}{B(x)}$, then $\chi\,a = b$ and we are done.
If $s\,a = \Right(f)$ where $|x|\le q$ and
$\ire{f}{(B(2x) \to B(x))}$, then 
$\chi_{n+1}\,a = \strictapp{f}{(\chi_n\,a')}$ for $a'=\multwo\,a$. It follows 
$\defined{\chi_{n}\,a'}$. By i.h., $\ire{(\chi\,a')}{B(2x)}$ and therefore
$\ire{(f\,(\chi\,a'))}{B(x)}$. But 
$f\,(\chi\,a') = \strictapp{f}{(\chi\,a')} = \chi\,a$
since $\defined{\chi\,a'}$.
\end{proof}
The proof above can also be viewed as an instance of 
Scott-induction with the admissible predicate 
$\lambda d\,.\,\forall a\,(\defined{d\,a} \to \ire{(\chi\,a)}{B(x)})$ 
applied to $d = \chi$. 

\paragraph{Remarks.} 
Thanks to the rule (Rest-stab), one can use classical logic for the right argument
of a restriction. The rules (Rest-bind) and (Rest-return) show that restriction 
behaves like a monad in its left argument. On the other hand, the concurrency 
operator does not enjoy this property since a corresponding bind-rule is missing. 
Instead, one has only the weaker rule (Conc-mp) which can also be seen as a 
monotonicity rule. This shortcoming will be addressed in Sect.~\ref{sec-mnd}.

\begin{lemma}[$\BT$]
\label{lem-restapart}
Elements $x$ of a constructive field satisfy $\rt{x\neq 0}{x\,\apart\,0}$.
%

The realizer extracted from the proof is the program $\varphi_1$ of 
Lemma~\ref{lem-apart}.
\end{lemma}
\begin{proof}
The proof is the same as for Lemma~\ref{lem-apart}, noticing that for the 
proof of the premise, the assumption $x\neq 0$ is not used.
\end{proof}

\emph{Remark.} In contrast to Lemma~\ref{lem-restapart}, 
the formula $\forall x\,(\A(x) \to (\rt{x=0}{x=0}))$ is 
not realizable (where $x=0$ is again an atomic nc formula). This can be seen
by a simple (domain-theoretic) continuity argument. If $\psi$ were
a realizer then, $\psi(\lambda n\,.\,0) = \Nil$, since $\lambda n\,.\,0$ 
realizes $\A(0)$ and $0=0$ is realizable i.e.\ holds. Since $\psi$ is continuous
there is some $k\in\NN$ such that $\psi(f) =\Nil$ where $f(n)=0$ if $n<k$
and $f(n)=2^{-(k+1)}$ if $n \ge k$. Clearly $f$ realizes $\A(2^{-(k+1)})$, but
$\Nil$ does not realize the false equation $2^{-(k+1)}=0$.

The following lemma is crucial for concurrent Gaussian elimination.
\begin{lemma}[$\BT$]
\label{lem-twopivot}
If elements $x,y$ of a constructive field are not both $0$, then,
concurrently, one of them is apart from $0$, that is,
$\Set(x\,\apart\,0 \lor y\,\apart\,0)$.

The extracted program is
\begin{eqnarray*}
\varphi_2(f,g) &=& \Amb(\strictapp{\Left}{(\varphi_1(f))},
                        \strictapp{\Right}{(\varphi_1(g))})\,.
\end{eqnarray*} 

\end{lemma}

\begin{proof}
By Lemma~\ref{lem-restapart} and the fact that restriction is monotone in its
left argument (which follows from the rule (Rest-bind)), we have
$\rt{x\neq 0}{(x\,\apart\,0 \lor y\,\apart\,0)}$ as well as
$\rt{y\neq 0}{(x\,\apart\,0 \lor y\,\apart\,0)}$.
Since not both $x$ and $y$ are $0$, we can apply rule (Rest-class-or).
%
%
\end{proof}
\emph{Remarks.} Intuitively, the extracted program $\varphi_2$ 
consists of two processes 
$\alpha=\strictapp{\Left}{(\varphi_1(f))}$ and 
$\beta=\strictapp{\Right}{(\varphi_1(g))}$ 
which search concurrently for approximations to $x$ and $y$ respectively, 
until $\alpha$ finds a 
$k$ where $|f(k)|>2^{-k}$ or $\beta$ finds an $l$ where 
$|g(l)|>2^{-l}$, returning the numbers with a corresponding flag $\Left$ or 
$\Right$. Both, or only one of the searches might 
be successful and which of the successful search results is taken is 
decided nondeterministically. 

In this particular situation, concurrency could be avoided (even without 
synchronization or interleaving): Since
not both $x$ and $y$ are $0$, $x^2+y^2$ is nonzero and hence apart
from $0$. Computing sufficiently good
approximations to $x^2+y^2$ and $x^2$ and comparing them, one can decide
whether $x$ or $y$ is apart from $0$.  However, to compute approximations
of $x^2+y^2$, both, $x$ and $y$ need to be equally well approximated, thus
the overall computation time is the \emph{sum} of the computation times for 
$x$ and $y$.
%
%
However, using concurrency the computation time will be the \emph{minimum}. 
Imagine $x$ and $y$ being very small positive reals and the realizer $f$ of
$F(x)$ providing very fast, but the realizer $g$ of $F(y)$ very slow
approximations.
Then the extracted concurrent program will terminate fast with a 
result $\Left(k)$ (computed by the process $\alpha$ searching $f$).
%
Hence, this example illustrates the fact that concurrent
realizability not only supports a proper treatment of partiality,
as in the case study on infinite Gray code~\citep{BergerTsuikiCFP},
but also enables us to exploit the efficiency gain of concurrency at the
logical level.
\subsection{Finitely threaded concurrency}
\label{sec-bnd}
$\CFP$ can be easily generalized to concurrency with an 
arbitrary finite number of threads. 
%
We generalize the operator $\Set$ to $\Set_n$, where $n$ is a positive 
natural number (so that $\Set$ corresponds to $\Set_2$), and allow the 
constructor $\Amb$ to take an arbitrary positive but finite number of arguments.
We define
that $a$ realizes $\Set_n(A)$ if 
$a$ is of the form $\amb(a_1,\ldots,a_m)$
with $1\le m\le n$
such that at least one $a_i$ is defined and all defined $a_i$ realize $A$:
\begin{eqnarray*}
a\,\re\,\Set_n(A) 
&\eqdef& \bigvee_{1\le m\le n}\,(a = \amb(a_1,\ldots,a_m) \land\\
&& \qquad \quad \bigvee_{1 \le i \le m} \defined{a_i} \land
           \bigwedge_{1 \le i \le m} (\defined{a_i} \to a_i\,\re\,A)) 
\end{eqnarray*}
%
%


The proof rules for $\Set_n$ are similar to those for $\Set$:
%

\[
  \infer[\hbox{Conc-class-or-n}]{
  \Set_{n}(A) }{
  \neg\neg \bigvee_{1\le i\le n}B(i) \qquad
   \bigwedge_{1\le i \le n} \rt{B(i)}{A} }
\]

\[
  \infer[\hbox{Conc-return-n}]{
  \Set_{n}(A) }{
  A}
\qquad
  \infer[\hbox{Conc-mp-n}]{
  \Set_{n}(B) }{
  A\to B \qquad \Set_{n}(A) }
\]

\begin{lemma}[$\BT$]
\label{lem-npivot}
If elements $x_1,\ldots, x_n$ of a constructive field are not all $0$, then,
$n$-concurrently, one of them is apart from $0$, that is,
$\Set_n(x_1\,\apart\,0 \lor \ldots \lor x_n\,\apart\,0)$.

Setting 
$\In_i \eqdef \lambda\,a\,.\,\Right^{i-1}(\Left(a))$ for $i\in\{1,\ldots,n\}$,
the extracted program is
\begin{eqnarray*}
\varphi_3(f_1,\ldots,f_n) &=& 
  \amb(\strictapp{\In_1}{(\varphi_1(f_1))},\ldots,
       \strictapp{\In_n}{(\varphi_1(f_n))})\,.
\end{eqnarray*} 

\end{lemma}

\begin{proof}
Similar to Lemma~\ref{lem-twopivot} using the rules for $\Set_n$.
%
\end{proof}
\emph{Remark.} We treated the index $n$ as a parameter on the meta-level,
i.e.\ $\Set_{n}$ is an operator for every $n$. 
On the other hand, in a stringent 
formalization, one has to treat $n$ as a formal parameter, 
i.e.\ $\Set_{n}(A)$ is a formula with a free variable $n$. Finitely iterated
conjunctions and disjunctions have to be formalized in a similar way.
This can be easily done, using formal inductive definitions, however we refrain
from carrying this out, since it would obscure
the presentation and impede the understanding of what follows.

\subsection{Monadic concurrency}
\label{sec-mnd}

The concurrency operators $\Set$ and $\Set_n$ were sufficient for the examples 
studied so far, since these did not involve iterated concurrent computation, 
i.e.\ we did not use the result of one concurrent computation as a parameter 
for another one. In other words, we did not use concurrency `in sequence'.
However, such sequencing will be required in our concurrent modelling
of Gaussian elimination (Sect.~\ref{sec-gauss}) since the choices of
previous pivot elements influence the choice of the next one.

The problem of sequencing becomes apparent if we nest 
non-determinism as defined by $\Set_n$ since this increases the
index $n$ bounding the number of processes. More precisely, the rule
\[
  \infer[\hbox{}]{
  \Set_{n^2}(A)
}{
\Set_n(\Set_n(A)\lor\False)
}
\]
is realizable, but this is optimal - we cannot lower the number $n^2$
in the conclusion to $n$.
Mathematically, the problem is that $\Set_n$ is not a monad.
However, we can turn $\Set_n$ into a monad by an inductive definition:
\[\Setomega_n(A) \eqmu \Set_n(A\lor\Setomega_n(A))\]
Clearly, $a\,\re\,\Setomega_n(A)$
holds iff $a$ is of the form $\amb(a_1,\ldots,a_m)$ for some $1\le m\le n$
such that at least one $a_i$ is defined and all defined $a_i$ realize 
$A\lor \Setomega_n(A)$.
%

Note that $\Setomega_n(A)$ is wellformed for arbitrary formulas $A$, not only
strict
ones.
\begin{lemma}
\label{lem-setstar}
The following rules follow from the rules for $\Set_n$ and are hence
realizable (for arbitrary formulas):
%
\[
  \infer[\hbox{Conc-return-n*}]{\ 
  \Setomega_n(A)
}{
A
}
\qquad
  \infer[\hbox{Conc-weak-bind}]{\
\Setomega_n(B)
}{
  A\to \Setomega_n(B)\ \ \  \Set_n(A) 
}
\]

\[
  \infer[\hbox{Conc-bind}]{\
\Setomega_n(B)
}{
  A\to \Setomega_n(B)\ \ \  \Setomega_n(A) 
}
\]
\end{lemma}
\begin{proof}
(Conc-return-n*) follows immediately from the definition of $\Setomega_n$.

To show (Conc-weak-bind), assume $A\to \Setomega_n(B)$ and $\Set_n(A)$.
Then also $A\to B \lor \Setomega_n(B)$ and hence, by (Conc-mp-n), 
$\Set_n(B \lor \Setomega_n(B))$. It follows $\Setomega_n(B)$, by the definition
of $\Setomega_n$.

To prove (Conc-bind), we assume $A\to \Setomega_n(B)$ and show
$\Setomega_n(A) \to \Setomega_n(B)$ by strictly positive induction
on the definition of $\Setomega_n(A)$. Hence, it suffices to show
$A \lor \Setomega_n(B) \to \Setomega_n(B)$, which follows from the assumption
$A\to \Setomega_n(B)$.
\end{proof}

We characterize realizability of $\Setomega_n(A)$ by term rewriting.
Call $\amb(\vec a)$ \emph{sound} if 
$\forall\, i\,(\defined{a_i}\to \exists\, b\,(a_i=\Left(b)\lor a_i =\Right(b))$.
Now define
\begin{eqnarray*}
a \to_n a' &\eqdef& 
\exists\, m \le n \exists\, \vec a = a_1,\ldots,a_m 
   \land a = \amb(\vec a)\ \hbox{sound}\land\\
 &&\qquad  \bigvee_i(\exists\, b\,a_i = a' = \Left(b) \lor a_i = \Right(a')) 
\end{eqnarray*}
\begin{lemma}
\label{lem-snrr} 
$a\,\re\,\Setomega_n(A)$ iff $a$ is reducible w.r.t. $\to_n$ and 
every reduction sequence $a \to_n a' \to_n \ldots$ terminates and ends with
some $\Left(b)$ such that $b\,\re\,A$.
\end{lemma}
\begin{proof}
``only if'': Induction on $a\,\re\,\Setomega_n(A)$. 
Assume $a\,\re\,\Setomega_n(A)$.
Then $a = \amb(\vec a)$ and $a$ is reducible. Let $a \to_n a' \to_n \ldots$
be a reduction sequence. Case $a_i=a'=\Left(b)$.
Then $a'$ is not further reducible and $b\,\re\,A$, hence we are done.
Case $a_i=\Right(a')$ and $a'\,\re\,\Setomega_n(A)$. 
Then the induction hypothesis applies.
``if'': Assume that $a$ is reducible and 
every reduction sequence $a \to_n a' \to_n \ldots$ terminates and ends with
some $\Left(b)$ where $b\,\re\,A$. 
This means that the relation $\to_n$ restricted to iterated reducts
of $a$ is wellfounded. We show $a \,\re\,\Setomega_n(A)$ by 
induction on $\to_n$.
Since $a$ is reducible, $a \to_n a'$ for some $a'$. Hence $a=\amb(\vec a)$
and either $a_i = a' = \Left(b)$ for some $b$, or else $a_i = \Right(a')$.
Hence the first part of the definition of $a \,\re\,\Setomega_n(A)$ holds.
For the second part assume $\defined{a_i}$. Since $\amb(\vec a)$ is sound,
either $a_i=\Left(b)$ or $a_i=\Right(b)$.
Assume $a_i=\Left(b)$. We have to show $b\,\re\,A$. Since 
$a \to_n \Left(b)$ it follows that $\Left(b)$ terminates 
a reductions sequence starting with $a$. Hence $b\,\re\,A$, by assumption.
Finally, assume $a_0=\Right(a')$. We have to show $a'\,\re\,\Setomega_n(A)$.
Since $a \to_n a'$ this holds by induction hypothesis.
\end{proof}

The characterization of $a\,\re\,\Setomega_n(A)$ given in Lemma~\ref{lem-snrr}
shows that $a$ can be viewed as a nondeterministic process requiring
at most $n$ processes running concurrently (or in parallel).
Each step $a\to_n a'$ represents such a non-deterministic computation.
The fact that \emph{every} reduction sequence will lead to a realizer of $A$
means that after each step all other computations running in parallel
may safely be abandoned since no backtracking will be necessary.


\section{Gaussian elimination}
\label{sec-gauss}

By a \emph{matrix} we mean a quadratic matrix with real coefficients.
A matrix is \emph{non-singular} if its columns are linearly
independent. 
Let $\creals$ be a constructive field of reals (see Sect.~\ref{sec-realnumbers}).
A matrix $A$ is called \emph{$\creals$ matrix}
if all the coefficients of $A$ are in $\creals$.
We prove that every non-singular $\creals$ matrix has an $\creals$ left inverse,
and extract from the proof a concurrent version of the 
well-known matrix inversion algorithm based on Gaussian elimination. 
The proof is concurrent, because the non-zero element of a column 
(the pivot) is computed concurrently.

We fix a dimension $n\in\NN$ and let $k,l$ range over $\NN_n \eqdef \{1,\ldots,n\}$,
the set of valid matrix indices.
For a matrix $A$ we let $A(k,l)$ be its entry at the $k$-th row and $l$-th
column. 
%
%
We denote matrix multiplication by $A \ast B$, i.e.\ 
\[(A\ast B)(k,l) = \sum_{i=1}^{n} A(k,i)\cdot B(i,l)\]
By $E$ we denote the unit matrix, i.e.\ 
\[E(k,l) \eqdef \left\{
    \begin{array}{ll}
       1 &\hbox{if $k=l$}\\
       0      &\hbox{otherwise}
    \end{array}
     \right.
\]
We call an $\creals$ matrix \emph{invertible} if there exists an $\creals$ matrix $B$ 
such that $B \ast A = E$.

To formalize the above, one may add sorts for vectors
and matrices and a function for accessing their entries.
Furthermore, a function symbol for matrix multiplication which
needs three arguments, two for the matrices to be multiplied and one for the dimension.
The formula for matrix multiplication can then be stated as an axiom 
since it is admissibly Harrop (see Sect.~\ref{sec-ifp}).
Furthermore, one postulates admissible Harrop axioms for the existence of explicitly 
defined matrices, for example,
\[ \forall n\in\NN \exists E \forall i,j\in \{1,\ldots,n\} 
   ((i = j \to E(i,j)=1) \land ( i\neq j \to E(i,j)=0)\,. \] 
The notion of linear independence can be expressed as a Harrop formula as well.
There are many alternative formalization which, essentially, lead to the same result,
in particular, the extracted program will be the same.

\begin{theorem}[Gaussian elimination]
\label{thm-gauss}
Every non-singular $\creals$ matrix is invertible. 
\end{theorem}
\begin{proof}
For matrices $A$, $B$ and $i\in\{0,\ldots,n\}$ we set
\[ A \eqc{i} B  \eqdef \forall k,l\,(l > i \to A(k,l) = B(k,l)) \]
($A$ and $B$ coincide on the columns $i+1,\ldots,n$). 
Hence $A \eqc{0} B$ means $A=B$ whereas $A\eqc{n} B$ always holds.
To prove the theorem, it suffices to show:
\paragraph{Claim.}
For all $i\in\{0,\ldots,n\}$, if $A$ is a non-singular $\creals$ matrix 
such that $A \eqc{i} E$, then $A$ is invertible.
\bigskip

We prove the Claim by induction on $i$.

\emph{In the proof, we need to bear in mind that, since we want to 
use concurrent pivoting 
(Lemma~\ref{lem-npivot}), and we iterate the argument, 
we can only hope to prove invertibility of $A$
$\omega$-concurrently, that is, we prove in fact
\[ \forall i\in\{0,\ldots,n\}\,\forall\,A\in \nsfmat\,
   (A \eqc{i} E \to \Setomega_n(\exists B\in \fmat\, B \ast A=E))\]
where $A \in \mnsfmat$ means that $A$ is an $n$-dimensional
(non-singular) $F$ matrix.} 
%

The base case, $i=0$, is easy, since the hypothesis $A \eqc{0} E$ means $A=E$.
Therefore, we can take $B \eqdef E$.

For the step, assume $i>0$ and let
$A\in \nsfmat$
such that $A \eqc{i} E$.
It suffices to show that there is
$C\in \nsfmat$
such that 
$C \ast A \eqc{i-1} E$. Since then
$C \ast A\in \nsfmat$,
and therefore, by induction hypothesis, we find
$B\in \fmat$,
with $B \ast (C \ast A) = E$. Hence $(B \ast C) \ast A = E$,
by the associativity of matrix multiplication.

\emph{Thanks to the rule (Conc-weak-bind), it suffices to find the matrix
$C$ concurrently, that is, it suffices to prove 
$\Set_n(\exists\, C \in F_{ns}^{n\times n}\,C\ast A \eqc{i-1} E)$, 
since the induction hypothesis implies
\[(\exists\,C\in F_{ns}^{n\times n}\,C\ast A \eqc{i-1} E) \to  
                \Setomega_n(\exists B\in F_{ns}^{n\times n}\, B \ast A=E).\]
}
Since $A$ is non-singular and $A \eqc{i} E$, it is not the case that 
$A(k,i)=0$ for all $k\le i$. Because otherwise the $n$-tuple 
$(\alpha_1,\ldots,\alpha_n)$ with $\alpha_l = 0$ for $l< i$,
$\alpha_i=1$ and $\alpha_l=-A(l,i)$ for $l>i$ would linearly combine the columns
of $A$ to the zero-vector, that is $\sum_{l=1}^{n} \alpha_l\cdot A(k,l)=0$
for all $k$.



        

%
Therefore, by Lemma~\ref{lem-npivot}, we find, concurrently,
$k_0\le i$ such that $A(k_0,i) \apart 0$.
Hence, there exists $\alpha \in F$ such that $A(k_0,i)\cdot \alpha = 1$. 
Define the matrix $C_1$ by
\[C_1(k,l) \eqdef \left\{
    \begin{array}{ll}
       \alpha &\hbox{if $k=l = k_0$}\\
       1      &\hbox{if $k=l \neq k_0$}\\
       0      &\hbox{otherwise}
    \end{array}
     \right.
\]
and set $A_1 \eqdef C_1 \ast A$ (multiplying the $k_0$-th row by $\alpha$). 
Clearly, $A_1 \eqc{i} E$ and $A_1(k_0,i) = 1$. 
Further, define
\[C_2(k,l) \eqdef \left\{
    \begin{array}{ll}
       1 &\hbox{if $k=l\not\in\{k_0,i\}$ or $(k,l)\in\{(k_0,i),(i,k_0)\}$}\\
       0 &\hbox{otherwise}
    \end{array}
     \right.
\]
and set $A_2 \eqdef C_2 \ast A_1$ (swapping rows $k_0$ and $i$). 
Clearly, $A_2 \eqc{i} E$ and $A_2(i,i) = 1$. 
Finally, define
\[C_3(k,l) \eqdef \left\{
    \begin{array}{ll}
       -A_2(k,i) &\hbox{if $k\neq i$ and $l=i$}\\ 
       1         &\hbox{if $k=l$}\\
       0         &\hbox{otherwise}
    \end{array}
     \right.
\]
and set $A_3 \eqdef C_3 \ast A_2$ (subtracting
the $A_2(k,i)$ multiple of row $i$ from row $k$, for each $k\neq i$). 
Clearly, $A_3 \eqc{i-1} E$, 
%
and $C \eqdef C_3 \ast C_2 \ast C_1$ is a non-singular $\creals$ matrix
since $C_1,C_2,C_3$ are. By associativity of matrix multiplication we
have $C \ast A = A_3 \eqc{i-1} E$.
This completes the proof of the claim and hence the theorem.
\end{proof}
The above proof is written in such a detail that its formalization
is straightforward and a program can be extracted. We did this
(by hand since $\CFP$ has yet to be implemented), with the
lazy functional programming language Haskell as target language.
In~\citep{BergerTsuikiCFP} it is shown how $\Amb$ can be interpreted
as concurrency using the Haskell library \verb|Control.Concurrent|.

Since Thm.~\ref{thm-gauss} is stated w.r.t.\ an arbitrary constructive field $F$,
the extracted program is polymorphic in (realizers of) $F$.
We tested the extracted matrix inversion program with respect to the
constructive fields 
$\QQ$ (rational numbers (\ref{eq-rat})), 
$\A$ (fast rational Cauchy sequences implemented as functions~(\ref{eq-approx})), 
and $\Aco$ (fast rational Cauchy sequences implemented as 
streams~(\ref{eq-approxco})),
running it with matrices some of whose
entries are nonzero but very close to zero and hard to compute~\cite{githubUB}.

The result was that $\QQ$ is hopeless since exact rational numbers become huge, 
in terms of the size of their representation as fractions, and thus 
computationally unmanageable.
$\A$ and $\Aco$ performed similarly with $\Aco$ being slightly faster 
since the stream
representation can exploit memoization effects and Haskell's laziness.

We also compared the extracted programs with a variant where concurrency
is explicitly modelled through interleaving Cauchy sequences or digit streams.
Here, one notices the effect that the slowest Cauchy sequence
(the one whose computation consumes most time) determines the overall run time, 
instead of the fastest as it is the case with true concurrency.


\section{Conclusion}
\label{sec-conclusion}

This paper aimed to demonstrate the potential of constructive mathematics
regarding program extraction, not only as a possibility in principle,
but as a viable method that is able to capture 
computational paradigms that are crucial in current programming
practice.
%
%
We used $\CFP$, an extension of a logic for program extraction
by primitives for partiality and concurrency, to extract a concurrent
program for matrix inversion based on Gaussian elimination.
Although the proof of the correctness of the logic is quite involved,
the extension is very simple and consists of just two new logical operators,
restriction, a strengthening of implication that is able to control partiality,
and an operator that permits computation with a fixed number of concurrent
threads. The case study on Gaussian elimination showed that proofs in
this system are very close to usual mathematical practice.

The system $\CFP$ uses the axiom $\BT$, a generalization of Brouwer's Thesis,
which can also be viewed as an abstract form of bar induction.
$\BT$ was used to turn the Archimedean axiom for real numbers into
an induction schema which in turn allowed us to prove that nonzero elements
of a constructive field are apart from zero. The usual proof of the
latter result requires Markov's principle and the axiom of
countable choice.
Moreover, the introduction rule for the concurrency operator is a form
of the law of excluded middle.
For these reasons, $\CFP$ has to be regarded as a semi-constructive system,
which begs the question to which extent our results can be regarded
constructively valid.  Classical logic could be partially avoided by
negative translation, however, its use has computational significance
since it allows us to undertake searches without having a constructive
bound for the search. Replacing these unbounded searches by algorithms that
use constructively obtained (worst-case) bounds may result in a dramatic loss of
efficiency since the computation of these bounds may be expensive and useless
since the search usually stops far earlier than predicted by the bound 
(see~\citep{Pattinson21} for examples).
It is an interesting problem to find a formal system that is fully
constructive but, at the same time, allows for efficient program extraction as 
in our semi-constructive system.


\section*{Acknowledgements}
  This work has received funding from the European Union's Horizon 2020 
  research and innovation programme under the Marie Sklodowska-Curie 
  grant agreement No 731143 (CID).
  We also acknowledge support by the JSPS Core-to-Core Program (A. Advanced
  Research Networks) and JSPS KAKENHI grant number 15K00015. 


\bibliographystyle{ws-rv-van}
\bibliography{refs}


\end{document}